\documentclass[journal,twoside,web]{ieeecolor}

\usepackage{generic}
\usepackage{cite}
\usepackage{amsmath,amssymb,amsfonts}
\usepackage{algorithm}
\usepackage{algpseudocode}

\usepackage{listings}
\usepackage{graphicx}
\usepackage{textcomp}
\usepackage{url}
\usepackage{hyperref}
\usepackage{graphicx}

\usepackage{epsfig} 
\usepackage{cite}
\usepackage{lcsys}
\usepackage[caption=false,font=footnotesize]{subfig}

\newtheorem{theorem}{Theorem}[section]

\newtheorem{proposition}[theorem]{Proposition}
\newtheorem{corollary}[theorem]{Corollary}
\newtheorem{definition}[theorem]{Definition}
\newtheorem{remark}[theorem]{Remark}

\usepackage{amsmath} 
\usepackage{amssymb}  
\usepackage{xcolor}
\usepackage{caption}
\usepackage{graphicx}
\usepackage{mathtools}
\usepackage{caption}
\usepackage{float}
\usepackage{tikz}
\usepackage{url}

\usepackage{pgfplots}
\usepackage{xcolor}
\usetikzlibrary{matrix,arrows,calc,positioning,shapes,decorations.pathreplacing}
\usepackage{graphicx}

\usepackage{amsmath} 

\usepackage{enumerate}
\usepackage[all,tips]{xy}
\SelectTips{cm}{11}

\usepackage{bm}

\usepackage{mathtools}
\usepackage{siunitx}

\begin{document}

\def\BibTeX{{\rm B\kern-.05em{\sc i\kern-.025em b}\kern-.08em
    T\kern-.1667em\lower.7ex\hbox{E}\kern-.125emX}}
\markboth{\journalname, VOL. XX, NO. XX, XXXX 2017}
{Author \MakeLowercase{\textit{et al.}}: Preparation of Papers for IEEE Control Systems Letters (August 2022)}

\title{Feasibility and Singularity in High-Order Safety-Critical Control for Quadrotor UAVs}

\author{Omayra Yago Nieto$^{1}$ and Leonardo Colombo$^{2}$
\thanks{$^1$ Omayra Yago Nieto is with Universidad Politécnica de Madrid, Spain. {\tt\small omayra.yago.nieto@alumnos.upm.es}}
\thanks{$^2$Centre for Automation and Robotics (CSIC-UPM), Ctra. M300 Campo Real, Km 0,200, Arganda del Rey - 28500 Madrid, Spain, {\tt\small leonardo.colombo@csic.es}.}
\thanks{The authors acknowledge financial support from Grants PID2022-137909NB-C21 funded by MCIN/AEI/10.13039/501100011033 and iRoboCity2030-CM, Robótica Inteligente para Ciudades Sostenibles (TEC-2024/TEC-62), funded by the Programas de Actividades I+D en Tecnologías en la Comunidad de Madrid.}}

\maketitle
\thispagestyle{empty}


\begin{abstract}
We study high-order safety-critical control of quadrotor teams
under bounded inputs and pairwise collision-avoidance constraints.
Squared-distance barriers may lose thrust effectiveness when the
relative displacement is orthogonal to the available thrust
directions, while nonsingular constraints may still be jointly
infeasible under shared bounds. We characterize both phenomena
through pairwise effectiveness and aggregate feasibility measures.
A torque-aware dynamic extension exposes attitude torques in a
fourth-order barrier and prevents the extended-input row from
vanishing under positive thrust. Gaussian processes directly learn
the fourth-order HOCBF residual, providing robust margins without
differentiating unknown perturbations. Under residual-bound and
persistent-feasibility assumptions, the resulting QP guarantees
collision avoidance and recovers the nominal input whenever it
satisfies the robust safety and actuator constraints.
\end{abstract}


\section{Introduction}

Quadrotor UAVs are intrinsically underactuated: they generate
translational force only along the body-fixed vertical
direction~\cite{lee2010geometric}. Consequently, their collision-avoidance
capability depends on the instantaneous attitude and relative
geometry.

Control barrier functions (CBFs) enforce safety through forward
invariance of prescribed safe sets~\cite{AmesTAC2017}, while high-order
and exponential CBFs address constraints whose input appears only after
repeated differentiation~\cite{nguyen2016exponential,xiao2022highorder}.
Recent work further addresses robustness, performance-critical
constructions, and loss of input effectiveness at singular
configurations~\cite{tan2022highorder}, compatibility of multiple CBFs
under input bounds~\cite{tan2022compatibility}, robust multiple-CBF
enforcement with constrained inputs~\cite{cortez2022robustmultiple},
and distributed CBF-based safe navigation~
\cite{mestres2024distributed}.

Recent quadrotor safety filters have also incorporated
higher-order vehicle dynamics. Harms et al.~\cite{harms2025safe}
use a third-order nonlinear model and composite CBFs with
recursive-feasibility analysis, while
\cite{tscholl2026fastbridge} enforces high-relative-degree
collision constraints through the full quadrotor dynamics
under actuator limits. In contrast, our focus is the
pairwise thrust singularity and the exact aggregate
feasibility of shared bounded inputs in multi-UAV collision
avoidance.

Gaussian-process (GP) models have been combined with CBF-based
safety filters for uncertainty quantification and pointwise
feasibility~\cite{castaneda2021pointwise},
including HOCBF residual
learning~\cite{aali2024learning}. Here, residual learning provides a uniform
robust margin, while we study the geometry-induced singularity and
aggregate bounded-input feasibility.

Our recent work~\cite{yagonieto2026safetycritical}, developed an aggregated robust HOCBF-QP for quadrotor teams driven by decentralized navigation
functions, with safety conditional on QP feasibility. Here we show
that infeasibility has two distinct sources: pairwise loss of
thrust effectiveness and joint incompatibility of nonsingular
inequalities sharing bounded inputs.

We first characterize the thrust singular set and introduce
an exact scalar aggregate feasibility margin for the shared
bounded-thrust constraints. We then use a torque-aware dynamic extension: thrust rate appears at third
order, whereas attitude torques first enter at fourth order, where they
prevent the pairwise extended-input row from vanishing under positive
thrust bounds. Finally, we extend the aggregate feasibility analysis to the dynamically extended inputs, learn the scalar fourth-order HOCBF residual with GPs to construct robust margins, and formulate a torque-aware safety-critical QP with
pairwise channel nonsingularity and local feasibility guarantees.

Section~II studies thrust-level feasibility; Sections~III--IV
develop the torque-aware robust controller and its local feasibility
conditions; Section~V presents simulations, and Section~VI concludes.

\section{Thrust-Based High-Order Safety and Aggregate Feasibility}\label{secII}

\subsection{Pairwise HOCBF safety and thrust effectiveness}

Consider $N$ quadrotor UAVs indexed by
$\mathcal V=\{1,\ldots,N\}$. The state of agent $i$ is
$x_i=(p_i,v_i,R_i,\Omega_i)$, with $p_i,v_i,\Omega_i\in
\mathbb R^3$ and $R_i\in SO(3)$, and its control input is
$(f_i,\tau_i)\in\mathbb R\times\mathbb R^3$. The dynamics are (see \cite{lee2010geometric}, for instance)
\begin{align}
    \dot p_i &=v_i,\qquad
    \dot R_i=R_i\widehat{\Omega}_i,                         \label{eq:quad_kin}\\
    m_i\dot v_i
    &=f_iR_ie_3+m_i g+\Delta_{f,i}(p_i,v_i),                    \label{eq:quad_trans}\\
    J_i\dot\Omega_i+\Omega_i\times J_i\Omega_i
    &=\tau_i+\Delta_{\tau,i}(x_i),                          \label{eq:quad_rot}
\end{align}
where $m_i>0$, $J_i\succ0$, $e_3=(0,0,1)^\top$,
$g=-g_0e_3$, \(g_0>0\) and $\widehat a\,b=a\times b$.

Throughout, the analysis is restricted to a compact operating
set. The force perturbations $\Delta_{f,i}(p_i,v_i)$ are
$C^2$, while $\Delta_{\tau,i}(x_i)$ are locally Lipschitz
there. The restricted dependence
$\Delta_{f,i}=\Delta_{f,i}(p_i,v_i)$ ensures that, after the
dynamic extension introduced in Section~\ref{sec:singularity_free}, the uncertain
fourth-order contribution modifies only the drift term and not
the extended-input coefficient. We consider locally Lipschitz feedback laws. Under
these assumptions, the closed-loop solutions are unique and
the derivatives entering the fourth-order HOCBF are well
defined classically along trajectories. The fourth-order construction is required because the attitude
torques first appear explicitly after four differentiations of
the pairwise distance barrier.

For the safety analysis, let $\mu_{f,i}(p_i,v_i)$ be a learned
estimate of $\Delta_{f,i}(p_i,v_i)$ obtained here from GP
regression~\cite{rasmussen2006gpml}, and define
$\bar\Delta_i:=\Delta_{f,i}-\mu_{f,i}$. Let $\bar d_{f,i}(p_i,v_i)$ denote the associated
confidence bound, so that $\|d_{f,i}(p_i,v_i)\|\le
\bar\Delta_i(p_i,v_i)$ on the operating set.

Collision avoidance is encoded by an undirected graph
$\mathcal G=(\mathcal V,\mathcal E)$. For $(i,j)\in\mathcal E$,
define
\begin{equation}
    h_{ij}:=\|p_i-p_j\|^2-d_{\min}^2,
    \qquad
    \mathcal C_{ij}:=\{x:h_{ij}(x)\ge0\},
    \label{eq:hij}
\end{equation}
where $x=\operatorname{col}(x_1,\ldots,x_N)$ and
$d_{\min}>0$.

Because the singularity studied below corresponds to a loss of uniform relative degree, we briefly recall the
least-relative-degree HOCBF construction of~\cite{tan2022highorder}. For
$\dot x=F(x)+G(x)u$, a function $h\in C^r$ has least
relative degree $r$ on $D$ if
$L_GL_F^kh=0$, $k=0,\ldots,r-2$, where $L_F$ denotes
the Lie derivative along $F$ and
$L_Gq:=\nabla q^\top G$. It has uniform relative degree
$r$ if also $L_GL_F^{r-1}h\neq0$ on $D$. With sufficiently differentiable extended class-$\mathcal K$
functions $\alpha_k$,
$\varphi_0:=h$ and
$\varphi_k:=\dot\varphi_{k-1}+\alpha_k(\varphi_{k-1})$,
a locally Lipschitz feedback satisfying $\varphi_r\ge0$
renders
$\bigcap_{k=0}^{r-1}\{\varphi_k\ge0\}$ forward invariant
whenever the HOCBF condition is feasible.

 Let $p_{ij}:=p_i-p_j$, $v_{ij}:=v_i-v_j$, and $f_{ij}:=\operatorname{col}(f_i,f_j)$. Since $\dot h_{ij}=2p_{ij}^\top v_{ij}$, the true second-order HOCBF defined by $\varphi_{0,ij}:=h_{ij}$, $\varphi_{1,ij}:=\dot h_{ij}+\alpha_1(h_{ij})$, and $\varphi_{2,ij}:=\dot\varphi_{1,ij}+\alpha_2(\varphi_{1,ij})$ admits the decomposition \begin{equation} \varphi_{2,ij} = \widehat c_{ij}(x) +B_{ij}(x)f_{ij} +\delta_{ij}^{(2)}(x), \label{eq:pairwise_hocbf_uncertain} \end{equation} where \begin{equation} B_{ij}(x) := 2\begin{bmatrix} \dfrac{p_{ij}^\top R_ie_3}{m_i} & -\dfrac{p_{ij}^\top R_je_3}{m_j} \end{bmatrix}, \label{eq:Bij} \end{equation} and $\widehat c_{ij} :=2\|v_{ij}\|^2 +2p_{ij}^\top(\mu_{f,i}/m_i-\mu_{f,j}/m_j) +\alpha_1'(h_{ij})\dot h_{ij} +\alpha_2(\dot h_{ij}+\alpha_1(h_{ij}))$. Moreover, $\delta_{ij}^{(2)} :=2p_{ij}^\top(d_{f,i}/m_i-d_{f,j}/m_j)$ satisfies \[ \left|\delta_{ij}^{(2)}\right|\leq 2\|p_{ij}\|
\left( \frac{\bar\Delta_i}{m_i} + \frac{\bar\Delta_j}{m_j}\right)=:\rho_{ij}^{(2)}.\] Here, $\widehat c_{ij}$ contains the known terms and the
GP-based force estimates, whereas $\delta_{ij}^{(2)}$ collects
the residual force uncertainty. Writing the input-dependent
term separately as $B_{ij}f_{ij}$ isolates the relative-thrust
coefficient whose norm is used below to quantify thrust
effectiveness. Hence a robust pairwise HOCBF condition is $B_{ij}(x)f_{ij}\ge-\widehat c_{ij}(x)+\rho_{ij}^{(2)}(x)$. 
 
 \begin{definition} The pairwise thrust-effectiveness measure and thrust singular set are $\gamma_{ij}^{\rm th}(x):=\|B_{ij}(x)\|$ and $\Sigma_{ij}^{\rm th}:=\{x:B_{ij}(x)=0\}$. \end{definition} 
 
 Equivalently,
$\Sigma_{ij}^{\rm th}
=\{x:p_{ij}^\top R_ie_3=0,\,
p_{ij}^\top R_je_3=0\}$.
Thus $h_{ij}$ has uniform relative degree two away from
$\Sigma_{ij}^{\rm th}$. At the singular set, the thrust
magnitudes cannot modify $\varphi_{2,ij}$ instantaneously.
The uncertainty margin shifts the feasibility level but leaves
the thrust singular set unchanged.
 
\subsection{Robust aggregate HOCBF feasibility}

Let $f:=\operatorname{col}(f_1,\ldots,f_N)$ and assume bounded
thrust inputs
$\displaystyle{\mathcal F:=\prod_{i=1}^N[f_i^{\min},f_i^{\max}]},\,
 0<f_i^{\min}<f_i^{\max}$.
Embedding $B_{ij}(x)$ into the full thrust vector, let
$\bar B_{ij}(x)\in\mathbb R^{1\times N}$ denote the row whose
only nonzero entries are the $i$th and $j$th components
given by~\eqref{eq:Bij}, and define
$r_{ij}^{\rm rob}(x):=
-\widehat c_{ij}(x)+\rho_{ij}^{(2)}(x)$. The robust pairwise HOCBF condition is therefore
\begin{equation}
    \bar B_{ij}(x)f\ge r_{ij}^{\rm rob}(x),
    \qquad (i,j)\in\mathcal E.
    \label{eq:pairwise_affine_constraint}
\end{equation}

Stacking all pairwise inequalities gives
\begin{equation}
    B(x)f\ge r^{\rm rob}(x),
    \qquad f\in\mathcal F,
    \label{eq:aggregate_constraints}
\end{equation}
where $B(x)\in\mathbb R^{M\times N}$,
$M:=|\mathcal E|$, and $r^{\rm rob}(x)\in\mathbb R^M$
stacks the corresponding $r_{ij}^{\rm rob}$. Robust aggregate
feasibility refers to the simultaneous satisfaction of all
pairwise inequalities by a single $f\in\mathcal F$, i.e.,
their pointwise compatibility under the shared input
constraints~\cite{tan2022compatibility}.

\begin{definition}
For a fixed state $x$, define
\begin{equation}
    \Gamma_{\rm th}^{\rm rob}(x)
    :=
    \max_{f\in\mathcal F}
    \min_{(i,j)\in\mathcal E}
    \bigl(
        \bar B_{ij}(x)f-r_{ij}^{\rm rob}(x)
    \bigr).
    \label{eq:Gamma_th_primal}
\end{equation}
Thus, $\Gamma_{\rm th}^{\rm rob}(x)$ is the largest common
margin achievable across all robust pairwise HOCBF
inequalities with a single admissible thrust vector.
\end{definition}

\begin{remark}
Pairwise thrust effectiveness and robust aggregate feasibility
are distinct properties. The quantity $\gamma_{ij}^{\rm th}$
measures the input effectiveness of a single edge, whereas
$\Gamma_{\rm th}^{\rm rob}$ measures simultaneous feasibility
under the shared bounded thrust inputs and uncertainty margins.
\end{remark}

For $y\in\mathbb R^N$, let
$\sigma_{\mathcal F}(y):=\max_{f\in\mathcal F}y^\top f$
denote the support function of $\mathcal F$~\cite{boyd2004convex}.

\begin{theorem}\label{thm:robust_aggregate_feasibility}
Consider the robust bounded-thrust HOCBF inequalities
\eqref{eq:aggregate_constraints}. Let
$\Delta_M:=\{\lambda\in\mathbb R_{\ge0}^M:
\mathbf 1^\top\lambda=1\}$. Then
\begin{equation}
    \Gamma_{\rm th}^{\rm rob}(x)
    =
    \min_{\lambda\in\Delta_M}
    \left[
        \sigma_{\mathcal F}\big(B(x)^\top\lambda\big)
        -
        \lambda^\top r^{\rm rob}(x)
    \right].
    \label{eq:Gamma_th_dual}
\end{equation}
Consequently, \eqref{eq:aggregate_constraints} is feasible
if and only if $\Gamma_{\rm th}^{\rm rob}(x)\ge0$.
Moreover, $\Gamma_{\rm th}^{\rm rob}(x)>0$ if and only if
there exist $f\in\mathcal F$ and $\varepsilon>0$ such that
$B(x)f\ge r^{\rm rob}(x)+\varepsilon\mathbf 1$.
\end{theorem}

\begin{proof}
For any $s\in\mathbb R^M$, every $\lambda\in\Delta_M$
defines a convex combination of its components, so
$\lambda^\top s\ge\min_k s_k$, with equality at
$\lambda=e_{k^\star}$ for
$k^\star\in\arg\min_k s_k$. Hence
$\min_k s_k
    =
    \min_{\lambda\in\Delta_M}\lambda^\top s$.

Applying this identity to
$s=B(x)f-r^{\rm rob}(x)$ in
\eqref{eq:Gamma_th_primal} gives
\[
    \Gamma_{\rm th}^{\rm rob}(x)
    =
    \max_{f\in\mathcal F}
    \min_{\lambda\in\Delta_M}
    \lambda^\top\big(B(x)f-r^{\rm rob}(x)\big).
\]
The sets $\mathcal F$ and $\Delta_M$ are compact and convex,
and the objective is continuous and affine in each variable.
Hence, by Sion's minimax theorem~\cite{sion1958general},
\[
    \Gamma_{\rm th}^{\rm rob}(x)
    =
    \min_{\lambda\in\Delta_M}
    \max_{f\in\mathcal F}
    \lambda^\top\big(B(x)f-r^{\rm rob}(x)\big).
\]
For fixed $\lambda$,
\[
    \max_{f\in\mathcal F}\lambda^\top B(x)f
    =
    \sigma_{\mathcal F}\big(B(x)^\top\lambda\big),
\]
which proves~\eqref{eq:Gamma_th_dual}.

By~\eqref{eq:Gamma_th_primal},
$\Gamma_{\rm th}^{\rm rob}(x)\ge0$ if and only if there exists
$f\in\mathcal F$ such that
$B(x)f\ge r^{\rm rob}(x)$.
Likewise, if $\Gamma_{\rm th}^{\rm rob}(x)>0$, an optimizer
$f^\star$ satisfies
$B(x)f^\star\ge r^{\rm rob}(x)
+\Gamma_{\rm th}^{\rm rob}(x)\mathbf 1$.
The converse follows directly from
\eqref{eq:Gamma_th_primal}.
\end{proof}

\medskip

\begin{corollary}
Let $\bar f:=\operatorname{col}(\bar f_1,\ldots,\bar f_N)$ and
$d^f:=\operatorname{col}(d_1^f,\ldots,d_N^f)$, where
$\bar f_i:=(f_i^{\max}+f_i^{\min})/2$ and
$d_i^f:=(f_i^{\max}-f_i^{\min})/2$. Then
\begin{equation}
    \Gamma_{\rm th}^{\rm rob}(x)
    =
    \min_{\lambda\in\Delta_M}
    \left[
        \bar f^\top B(x)^\top\lambda
        +
        (d^f)^\top|B(x)^\top\lambda|
        -
        \lambda^\top r^{\rm rob}(x)
    \right].
    \label{eq:Gamma_th_box}
\end{equation}
\end{corollary}

\begin{proof}
Let $y:=B(x)^\top\lambda$. Since $\mathcal F$ is a Cartesian
product of intervals,
\[
    \sigma_{\mathcal F}(y)
    =
    \sum_{i=1}^N
    \max_{f_i\in[f_i^{\min},f_i^{\max}]}y_i f_i
    =
    \bar f^\top y+(d^f)^\top|y|.
\]
Substituting $y=B(x)^\top\lambda$ into
\eqref{eq:Gamma_th_dual} yields~\eqref{eq:Gamma_th_box}.
\end{proof}

\section{Torque-Aware HOCBF Control}
\label{sec:singularity_free}

\subsection{Torque-aware dynamic extension}\label{secIIIB}

To identify the extended input channel without differentiating
the uncertain terms in
\eqref{eq:quad_trans}--\eqref{eq:quad_rot},
consider the reference dynamics obtained by omitting these
terms. Uncertainty is reintroduced at the HOCBF level below.
Let $h_{ij,\mathrm{ref}}^{(k)}$ denote the $k$th time
derivative of $h_{ij}$ along the reference dynamics. This
notation distinguishes reference-dynamics quantities from the
GP-based quantities denoted by hats in Section~\ref{secII}.

Let $\zeta_i:=R_ie_3$. From~\eqref{eq:quad_kin},
$\dot\zeta_i=-R_i\widehat e_3\Omega_i$.
Introduce $q_i:=\dot f_i$ and
$\nu_i:=\dot q_i=\ddot f_i$, and define
$z_i:=(p_i,v_i,R_i,\Omega_i,f_i,q_i)$, $z:=\operatorname{col}(z_1,\ldots,z_N)$, $w_i:=\operatorname{col}(\nu_i,\tau_i)$, and
$w_{ij}:=\operatorname{col}(\nu_i,\nu_j,\tau_i,\tau_j)$.

Since gravity cancels from the relative dynamics, let
$a_{ij}:=f_i\zeta_i/m_i-f_j\zeta_j/m_j$. Then
\[
h_{ij,\mathrm{ref}}^{(3)}
=6v_{ij}^\top a_{ij}+2p_{ij}^\top\dot a_{ij},
\,
h_{ij,\mathrm{ref}}^{(4)}
=6a_{ij}^\top a_{ij}
+8v_{ij}^\top\dot a_{ij}
+2p_{ij}^\top\ddot a_{ij}.
\]
For $a_i:=f_i\zeta_i/m_i$,
\[
\dot a_i
=\frac{q_i}{m_i}\zeta_i+\frac{f_i}{m_i}\dot\zeta_i,
\qquad
\ddot a_i
=\frac{\nu_i}{m_i}\zeta_i
+\frac{2q_i}{m_i}\dot\zeta_i
+\frac{f_i}{m_i}\ddot\zeta_i.
\]
Using the reference rotational dynamics,
$\ddot\zeta_i=d_{\zeta,i}
-R_i\widehat e_3J_i^{-1}\tau_i$, where
$d_{\zeta,i}:=
-R_i\widehat\Omega_i\widehat e_3\Omega_i
+R_i\widehat e_3J_i^{-1}
(\Omega_i\times J_i\Omega_i)$. Consequently,
\begin{equation}
    h_{ij,\mathrm{ref}}^{(4)}
    =
    \ell_{ij,\mathrm{ref}}^{(4)}(z)
    +
    A_{ij}^{\mathrm{ext}}(z)w_{ij}.
    \label{eq:h4_affine}
\end{equation}
where $\ell_{ij,\mathrm{ref}}^{(4)}$ collects the reference drift terms and
\begin{equation}
\begin{aligned}
A_{ij}^{\rm ext}(z)
=
2\Bigg[
&\frac{p_{ij}^{\top}R_ie_3}{m_i},
-\frac{p_{ij}^{\top}R_je_3}{m_j},\\
&-\frac{f_i}{m_i}
 p_{ij}^{\top}R_i\widehat e_3J_i^{-1},
\frac{f_j}{m_j}
 p_{ij}^{\top}R_j\widehat e_3J_j^{-1}
\Bigg].
\end{aligned}
\label{eq:Aext}
\end{equation}
For quantitative input-effectiveness measures, let
$S_{ij}:=\operatorname{diag}
(s_{\nu,i},s_{\nu,j},s_{\tau,i}I_3,s_{\tau,j}I_3)\succ0$
contain fixed characteristic input scales and define
$\gamma_{ij}^{\rm ext,S}(z):=
\|A_{ij}^{\rm ext}(z)S_{ij}\|$.

 We use the same notation $\mathcal C_{ij}$ for its natural lift to the extended state space. We next show that this
extended input channel does not vanish at a thrust singularity.

\begin{theorem}
\label{thm:singularity_removal}
Let $(i,j)\in\mathcal E$ and suppose
$z\in\mathcal C_{ij}$. If at least one of $f_i$ or $f_j$
is nonzero, then $A_{ij}^{\rm ext}(z)\neq0$. In particular,
if $x\in\Sigma_{ij}^{\rm th}$ and $f_i\neq0$, then
\begin{equation}
    \gamma_{ij}^{\rm ext,S}(z)
\ge
\frac{2|f_i|s_{\tau,i}}
{m_i\lambda_{\max}(J_i)}
\|p_{ij}\|>0.
    \label{eq:gamma_ext_bound}
\end{equation}
and analogously for agent $j$. Moreover, on the admissible
thrust domain $f_i\in[f_i^{\min},f_i^{\max}]$,
\begin{equation}
\gamma_{ij}^{\rm ext,S}(z)
\ge
\frac{2d_{\min}}{m_i}
\min\left\{
s_{\nu,i},
\frac{f_i^{\min}s_{\tau,i}}
{\lambda_{\max}(J_i)}
\right\}
>0.
\label{eq:gamma_ext_uniform_bound}
\end{equation}

Consequently,
$h_{ij}$ has uniform relative degree four
with respect to $w_{ij}$ along the reference extended dynamics
whenever at least one vehicle has nonzero thrust.
\end{theorem}

\begin{proof}
If $B_{ij}(x)\neq0$, the first two components of
$A_{ij}^{\rm ext}$ in~\eqref{eq:Aext} cannot both vanish.

Consider $x\in\Sigma_{ij}^{\rm th}$. Since $z\in\mathcal C_{ij}$,
$\|p_{ij}\|\ge d_{\min}>0$.
Let $q_i^p:=R_i^\top p_{ij}$. Then
$\|q_i^p\|=\|p_{ij}\|$ and, by
$x\in\Sigma_{ij}^{\rm th}$,
$(q_i^p)^\top e_3=0$. Hence
$\|(q_i^p)^\top\widehat e_3\|
    =
    \|q_i^p\|
    =
    \|p_{ij}\|$. Since $J_i\succ0$,
\[
\begin{aligned}
\left\|
p_{ij}^\top R_i\widehat e_3J_i^{-1}
\right\|
\ge
\sigma_{\min}(J_i^{-1})
\left\|
(q_i^p)^\top\widehat e_3
\right\| =
\frac{\|p_{ij}\|}{\lambda_{\max}(J_i)}.
\end{aligned}
\]
Thus the $i$th scaled torque block of
$A_{ij}^{\rm ext}S_{ij}$ has norm at least
$\frac{2|f_i|s_{\tau,i}}{m_i\lambda_{\max}(J_i)}
\|p_{ij}\|$, and~\eqref{eq:gamma_ext_bound} follows.
The argument for agent $j$ is identical. 

For the uniform bound, write
$q_i^p=(q_{i,1}^p,q_{i,2}^p,q_{i,3}^p)^\top$.
The $i$th thrust-acceleration and torque blocks give
\[
\begin{aligned}
\|A_{ij}^{\rm ext}S_{ij}\|^2
&\ge
\frac{4}{m_i^2}
\left[
s_{\nu,i}^2(q_{i,3}^p)^2
+
\frac{f_i^2s_{\tau,i}^2}
{\lambda_{\max}^2(J_i)}
\left((q_{i,1}^p)^2+(q_{i,2}^p)^2\right)
\right] \\
&\ge
\frac{4d_{\min}^2}{m_i^2}
\min\left\{
s_{\nu,i}^2,
\frac{(f_i^{\min})^2s_{\tau,i}^2}
{\lambda_{\max}^2(J_i)}
\right\},
\end{aligned}
\]
which proves~\eqref{eq:gamma_ext_uniform_bound}. Since $w_{ij}$
does not enter $h_{ij}^{(k)}$
for $k\le 3$, while the coefficient of $w_{ij}$
in $h_{ij}^{(4)}$
is $A_{ij}^{\rm ext}\neq 0$, the relative-degree claim follows.
\end{proof}

\begin{remark}
At a thrust singularity, the thrust-acceleration coefficients
vanish but the torque coefficients remain nonzero because
torques rotate the thrust directions. Thus the dynamic
extension changes the safety input channel rather than merely
differentiating the second-order HOCBF.
\end{remark}

\subsection{Robust extended feasibility and QP}

For the true dynamics, define, for $k=1,\ldots,4$
\begin{equation}
    \eta_{0,ij}:=h_{ij},\,
    \eta_{k,ij}:=
    \dot\eta_{k-1,ij}
    +\alpha_k(\eta_{k-1,ij}),
    \label{eq:eta_recursion}
\end{equation}
Let $\eta_{4,ij}^{\mathrm{ref}}$ denote the corresponding
fourth-order HOCBF expression along the reference dynamics.
By~\eqref{eq:h4_affine},
\begin{equation}
    \eta_{4,ij}^{\mathrm{ref}}(z,w_{ij})
    =
    \ell_{ij}^{\mathrm{ref}}(z)
    +
    A_{ij}^{\mathrm{ext}}(z)w_{ij},
    \label{eq:eta4_affine}
\end{equation}
where $\ell_{ij}^{\mathrm{ref}}$ collects the reference drift and
lower-order HOCBF terms.

For the true dynamics,
\[
\frac{d^2}{dt^2}\!\left(
\frac{f_i\zeta_i+\Delta_{f,i}}{m_i}\right)
=
\frac{\nu_i\zeta_i+2q_i\dot\zeta_i
+f_i\ddot\zeta_i+\ddot\Delta_{f,i}}{m_i},
\]
where $\ddot\zeta_i
=
d_{\zeta,i}
-R_i\widehat e_3J_i^{-1}
\big(\tau_i+\Delta_{\tau,i}\big)$. Moreover, since $\Delta_{f,i}=\Delta_{f,i}(p_i,v_i)\in C^2$, $\dot\Delta_{f,i}
=
D_p\Delta_{f,i}v_i
+D_v\Delta_{f,i}\dot v_i$, and $\ddot\Delta_{f,i}$ depends on
$(p_i,v_i,R_i,\Omega_i,f_i,q_i)$ but not on
$w_i=(\nu_i,\tau_i)$. Hence the uncertain terms modify only
the drift with respect to $w$, and the true and reference
fourth-order HOCBF expressions have the same input coefficient:
\[
\eta_{4,ij}(z,w_{ij})
=
\ell^{\rm true}_{ij}(z)+A^{\rm ext}_{ij}(z)w_{ij}.
\]
Define
$\delta_{ij}^{(4)}(z):=
\ell_{ij}^{\rm true}(z)-\ell_{ij}^{\mathrm{ref}}(z)$
and assume $|\delta_{ij}^{(4)}(z)|
\le
\rho_{ij}^{(4)}(z)$ on the prescribed operating domain. Hence a robust
$4^{th}$-order HOCBF condition is
$A_{ij}^{\rm ext}(z)w_{ij}
    \ge
    \rho_{ij}^{(4)}(z)-\ell_{ij}^{\mathrm{ref}}(z)$.

Since $f_i$ is an extended state, its physical bounds are enforced
through auxiliary second-order HOCBFs. Define
$g_{i,-}:=f_i-f_i^{\min},\,
g_{i,+}:=f_i^{\max}-f_i$, and $\vartheta_{1,i}^{\pm}
:=\dot g_{i,\pm}+\alpha_{f,1}(g_{i,\pm}),\,
\vartheta_{2,i}^{\pm}
:=\dot\vartheta_{1,i}^{\pm}
+\alpha_{f,2}(\vartheta_{1,i}^{\pm})$. Since $\dot f_i=q_i$ and $\dot q_i=\nu_i$,
$\vartheta_{2,i}^{\pm}$ are affine in $\nu_i$. Let $\mathcal W_0$ be a nonempty compact
convex set of admissible extended inputs, possibly encoding
coupled actuator limits, and define
\[
\mathcal W_a(z):=
\left\{
w\in\mathcal W_0:
\vartheta_{2,i}^{-}(z,w_i)\ge0,
\vartheta_{2,i}^{+}(z,w_i)\ge0,\, i\in\mathcal V
\right\}.
\]
For each fixed $z$, $\mathcal W_a(z)$ is convex and compact
whenever nonempty.

Let $w:=\operatorname{col}(w_1,\ldots,w_N)$. Embed $A_{ij}^{\rm ext}$ into the full input
space as $\bar A_{ij}^{\rm ext}$ and define
$r_{{\rm ext},ij}^{\rm rob}
:=\rho_{ij}^{(4)}-\ell_{ij}^{\mathrm{ref}}$.
Stacking the robust inequalities gives
\begin{equation}
    A^{\rm ext}(z)w
    \ge
    r_{\rm ext}^{\rm rob}(z),
    \qquad
    w\in\mathcal W_a(z),
    \label{eq:aggregate_ext}
\end{equation}
where $A^{\rm ext}$ stacks the rows $\bar A_{ij}^{\rm ext}$.

In analogy with~\eqref{eq:Gamma_th_primal}, define
\begin{equation}
    \Gamma_{\rm ext}^{\rm rob}(z)
    :=
    \max_{w\in\mathcal W_a(z)}
    \min_{(i,j)\in\mathcal E}
    \left(
        \bar A_{ij}^{\rm ext}(z)w
        -
        r_{{\rm ext},ij}^{\rm rob}(z)
    \right).
    \label{eq:Gamma_ext}
\end{equation}

For $y$ in the extended input space, let
$\sigma_{\mathcal W_a(z)}(y):=
\max_{w\in\mathcal W_a(z)} y^\top w$.

\begin{proposition}
\label{prop:extended_feasibility}
Suppose $\mathcal W_a(z)$ is nonempty. Let
$M:=|\mathcal E|$ and
$\Delta_M:=\{\lambda\in\mathbb R_{\ge0}^M:\mathbf 1^\top\lambda=1\}$.
Then
\begin{equation}
\Gamma_{\rm ext}^{\rm rob}(z)
=
\min_{\lambda\in\Delta_M}
\left[
    \sigma_{\mathcal W_a(z)}
    \big(A^{\rm ext}(z)^\top\lambda\big)
    -
    \lambda^\top r_{\rm ext}^{\rm rob}(z)
\right].
\label{eq:Gamma_ext_dual}
\end{equation}
Consequently, the robust fourth-order aggregate HOCBF
constraints~\eqref{eq:aggregate_ext} are feasible if and only if
$\Gamma_{\rm ext}^{\rm rob}(z)\ge0$.
\end{proposition}

\begin{proof}
The result follows from the minimax argument of
Theorem~\ref{thm:robust_aggregate_feasibility}, with
$\mathcal F$, $B$, and $r^{\rm rob}$ replaced by
$\mathcal W_a(z)$, $A^{\rm ext}$, and $r_{\rm ext}^{\rm rob}$,
respectively.
\end{proof}

Let $w^{\rm nom}(z)\in\mathcal W_0$ be a locally Lipschitz
nominal extended input. The robust torque-aware safety-critical controller is
\begin{equation}
\begin{aligned}
    w^\star
    =\arg\min_{w\in\mathcal W_a(z)}\quad&
    \frac12\|w-w^{\rm nom}\|_{H}^{2}\\
    {\rm s.t.}\quad&
    A^{\rm ext}(z)w\ge r_{\rm ext}^{\rm rob}(z),
\end{aligned}
\label{eq:extended_qp}
\end{equation}
where $H\succ0$. The applied inputs satisfy
$\dot f_i=q_i$, $\dot q_i=\nu_i^\star$, while
$\tau_i^\star$ is applied to~\eqref{eq:quad_rot}.

Let $\mathcal C:=\bigcap_{(i,j)\in\mathcal E}\mathcal C_{ij}$ and
$\mathcal C_{\rm H}:=
\bigcap_{(i,j)\in\mathcal E}
\bigcap_{k=0}^{3}\{z:\eta_{k,ij}(z)\ge0\}$.
Since $\eta_{0,ij}=h_{ij}$,
$\mathcal C_{\rm H}\subset\mathcal C$. Define also
$\displaystyle{
\mathcal C_f:=
\bigcap_{i\in\mathcal V}
\bigcap_{\sigma\in\{-,+\}}
\left\{
z:g_{i,\sigma}(z)\ge0,\;
\vartheta_{1,i}^{\sigma}(z)\ge0
\right\}}$.

\begin{theorem}
\label{thm:sf_safety}
Consider the dynamically extended quadrotor system and suppose
$z(0)\in\mathcal C_{\rm H}\cap\mathcal C_f$.
Assume that the fourth-order residual bound holds for every
$(i,j)\in\mathcal E$ on the prescribed operating domain,
\eqref{eq:extended_qp} is feasible for all $t\geq0$, and its
solution is locally Lipschitz. Then
$\mathcal C_{\rm H}\cap\mathcal C_f$ is forward invariant,
collision avoidance is guaranteed, and
$f_i(t)\in[f_i^{\min},f_i^{\max}]$ for all
$i\in\mathcal V$. Moreover, if
$w^{\rm nom}(z)\in\mathcal W_a(z)$ satisfies the robust
collision constraints, then $w^\star=w^{\rm nom}$.
\end{theorem}

\begin{proof}
For every $(i,j)\in\mathcal E$, the QP constraint gives
$\eta_{4,ij}^{\mathrm{ref}}(z,w_{ij}^{\star})
\geq\rho_{ij}^{(4)}(z)$.
Since
$\eta_{4,ij}=\eta_{4,ij}^{\mathrm{ref}}+\delta_{ij}^{(4)}$
and $|\delta_{ij}^{(4)}|\leq\rho_{ij}^{(4)}$, we obtain
$\eta_{4,ij}\geq0$.

Therefore
$\dot\eta_{3,ij}\ge-\alpha_4(\eta_{3,ij})$.
Since $\eta_{3,ij}(0)\ge0$, the comparison principle gives
$\eta_{3,ij}(t)\ge0$. Repeating recursively yields
$\eta_{2,ij}(t)\ge0$, $\eta_{1,ij}(t)\ge0$, and
$\eta_{0,ij}(t)=h_{ij}(t)\ge0$.
Hence $z(t)\in\mathcal C_{\rm H}$ for all $t\ge0$.
Since $\mathcal C_{\rm H}\subset\mathcal C$, collision
avoidance follows.

Moreover, since $w^\star\in\mathcal W_a(z)$,
$\vartheta_{2,i}^{\pm}\ge0$ for every $i\in\mathcal V$.
Hence
$\dot\vartheta_{1,i}^{\pm}
\ge-\alpha_{f,2}(\vartheta_{1,i}^{\pm})$,
and, since $\vartheta_{1,i}^{\pm}(0)\ge0$, the comparison
principle gives $\vartheta_{1,i}^{\pm}(t)\ge0$.
Therefore
$\dot g_{i,\pm}\ge-\alpha_{f,1}(g_{i,\pm})$,
and $g_{i,\pm}(0)\ge0$ implies $g_{i,\pm}(t)\ge0$.
Thus $\mathcal C_f$ is forward invariant and
$f_i(t)\in[f_i^{\min},f_i^{\max}]$ for all $i$.

Finally, if $w^{\rm nom}(z)\in\mathcal W_a(z)$ satisfies the
robust collision constraints, it is feasible for
\eqref{eq:extended_qp}. Since the objective is strictly convex
and is minimized at $w=w^{\rm nom}$, we obtain
$w^\star=w^{\rm nom}$.
\end{proof}

\begin{remark}

The common true/reference input coefficient
together with Theorem~\ref{thm:singularity_removal} ensures a nonzero pairwise extended-input row, whereas Proposition~\ref{prop:extended_feasibility} shows that aggregate
feasibility remains a distinct bounded-input property.
Thus, several nonsingular fourth-order constraints may
still be jointly infeasible.
\end{remark}

\section{Learning-Based Robust Margins and Local Feasibility}
\label{sec:robust_singularity_free_safety}
Section~\ref{sec:singularity_free} incorporates a uniform fourth-order residual
bound into the robust HOCBF constraints and QP. We now construct this bound from data and quantify local feasibility.

At second order, learned force estimates and uncertainty bounds
provide $\mu_{f,i}$ and $\bar\Delta_i$ in Section~\ref{secII}. At fourth order,
propagating them through successive derivatives may require derivatives
of the unknown perturbations, so we instead learn directly the scalar
residual $\delta_{ij}^{(4)}$ from Section~\ref{sec:singularity_free}.

For each $(i,j)\in\mathcal E$, let
$z_{ij}:=\operatorname{col}(z_i,z_j)$. By the pairwise
structure above, $\delta_{ij}^{(4)}$ depends only on $z_{ij}$.
We model $\delta_{ij}^{(4)}(z_{ij})$ by a GP with posterior
mean $\mu_{ij}^{\delta}(z_{ij})$ and standard deviation
$\sigma_{ij}^{\delta}(z_{ij})$~\cite{rasmussen2006gpml}. Thus, the GP
input dimension is pairwise and does not grow with the number of agents.
For prescribed $\beta_{ij}>0$, assume the event
\[
\mathcal E_{\rm GP}:=
\left\{
\left|
\delta_{ij}^{(4)}(z_{ij})-\mu_{ij}^{\delta}(z_{ij})
\right|
\le
\sqrt{\beta_{ij}}\sigma_{ij}^{\delta}(z_{ij}),
\ \forall (i,j)\in\mathcal E
\right\}
\]
holds uniformly on the prescribed operating domain. Such bounds can be
obtained under standard GP concentration assumptions~\cite{srinivas2010gpucb};
here $\mathcal E_{\rm GP}$ is an explicit assumption of the safety analysis.
Define $\rho_{ij}^{(4)}(z):=
|\mu_{ij}^{\delta}(z_{ij})|
+\sqrt{\beta_{ij}}\sigma_{ij}^{\delta}(z_{ij})$.
Then $|\delta_{ij}^{(4)}(z)|\le\rho_{ij}^{(4)}(z)$
for every $(i,j)\in\mathcal E$, and
Theorem~\ref{thm:sf_safety} applies on this event.

For the quantitative feasibility analysis below, let $S\succ0$
be the global diagonal scaling whose pairwise restriction is
$S_{ij}$ defined in Section~\ref{secIIIB}, and set
$\bar w=S^{-1}w$. Then
$A^{\rm ext}w=(A^{\rm ext}S)\bar w$. Since $S$ is nonsingular,
normalization does not change whether an input-effectiveness
row vanishes, but makes its norm dimensionally meaningful.

\subsection{Local feasibility}

Theorem~\ref{thm:sf_safety} provides a conditional safety guarantee under
persistent feasibility of~\eqref{eq:extended_qp}. We now give pointwise sufficient
conditions for feasibility of the robust constraints.

For $(i,j)\in\mathcal E$, define the nominal robust
deficit $r_{ij}(z)
:=
\left[
\rho_{ij}^{(4)}(z)
-
\eta_{4,ij}^{\mathrm{ref}}(z,w_{ij}^{\mathrm{nom}})
\right]_{+}$. Set $A_{ij}^{\rm ext,S}:=A_{ij}^{\rm ext}S_{ij}$. Let $\mathcal W_{a,ij}(z)$ denote the projection of
$\mathcal W_a(z)$ onto $w_{ij}$, and define
$\bar{\mathcal W}_{a,ij}(z):=
S_{ij}^{-1}\mathcal W_{a,ij}(z)$.

\begin{proposition}
\label{prop:single_edge_feasibility}
Suppose
$\mathcal B(\bar w^{\rm nom}_{ij},r_w)\subset
\bar{\mathcal W}_{a,ij}(z)$, where
$\bar w^{\rm nom}_{ij}:=S_{ij}^{-1}w^{\rm nom}_{ij}$
and $r_w>0$. If
$r_{ij}(z)<r_w\|A_{ij}^{\rm ext,S}(z)\|$, then the robust
constraint for edge $(i,j)$ is feasible.
\end{proposition}

\begin{proof}
If $r_{ij}=0$, the nominal input is feasible. Otherwise,
choose
$\delta\bar w_{ij}
=\lambda(A_{ij}^{\rm ext,S})^\top/\|A_{ij}^{\rm ext,S}\|$,
with
$r_{ij}/\|A_{ij}^{\rm ext,S}\|<\lambda<r_w$.
Then
$A_{ij}^{\rm ext,S}\delta\bar w_{ij}
=\lambda\|A_{ij}^{\rm ext,S}\|>r_{ij}$,
while $\bar w^{\rm nom}_{ij}+\delta\bar w_{ij}
\in\bar{\mathcal W}_{a,ij}(z)$.
\end{proof}

Let $A_{\rm ext}^S(z)$ stack the rows
$\bar A_{ij}^{\rm ext}(z)S$, $(i,j)\in\mathcal E$, and let
$r(z)$ stack the corresponding deficits $r_{ij}(z)$. Define
$\bar{\mathcal W}_a(z):=S^{-1}\mathcal W_a(z)$.

\begin{proposition}
\label{prop:multi_edge_feasibility}
Suppose
$\mathcal B(\bar w^{\rm nom},r_w)\subset\bar{\mathcal W}_a(z)$,
where $\bar w^{\rm nom}:=S^{-1}w^{\rm nom}$ and $r_w>0$.
If $A_{\rm ext}^S(z)$ has full row rank and
$\|(A_{\rm ext}^S(z))^\dagger r(z)\|<r_w$, then \eqref{eq:extended_qp} is feasible at $z$.
\end{proposition}

\begin{proof}
Choose $\delta\bar w=(A_{\rm ext}^S)^\dagger r$.
Full row rank gives $A_{\rm ext}^S\delta\bar w=r$,
so all violated robust inequalities are satisfied. The norm
condition ensures
$\bar w^{\rm nom}+\delta\bar w\in\bar{\mathcal W}_a(z)$.
\end{proof}

These conditions are pointwise and do not imply recursive
feasibility; accordingly, Theorem~\ref{thm:sf_safety} remains conditional on
persistent feasibility of~\eqref{eq:extended_qp}. Since
$A_{\rm ext}^S(z)\in\mathbb R^{M\times 4N}$, full row rank
requires $M=|\mathcal E|\le4N$, limiting
Proposition~\ref{prop:multi_edge_feasibility} for dense graphs.

\subsection{Decentralized implementation}

Assume $\mathcal W_0=\prod_{i=1}^N\mathcal W_{0,i}$ where each $\mathcal W_{0,i}\subset\mathbb R^4$ may encode
coupled actuator limits; the product assumption is only across agents. Define
\[
\mathcal W_{a,i}(z_i):=
\{w_i\in\mathcal W_{0,i}:
\vartheta_{2,i}^{-}(z_i,w_i)\ge0,\;
\vartheta_{2,i}^{+}(z_i,w_i)\ge0\}.
\]
Then $\mathcal W_a(z)=\prod_{i=1}^N\mathcal W_{a,i}(z_i)$.
Let $\mathcal N_i$ denote the neighbors of agent $i$.
Using a reference input $w_j^{\rm ref}$ for each neighbor,
agent $i$ computes
\[
\begin{aligned}
w_i^\star
&=\arg\min_{w_i\in\mathcal W_{a,i}(z_i)}\quad
\frac12\|w_i-w_i^{\rm nom}\|_{H_i}^2\\
{\rm s.t.}\quad&
\eta_{4,ij}^{\mathrm{ref}}(z,w_i,w_j^{\mathrm{ref}})
\ge
\rho_{ij}^{(4)}(z)+\rho_{ij}^{\rm comm}(z),
\quad j\in\mathcal N_i,
\end{aligned}
\]
where $H_i\succ0$. The reference may be a synchronized,
communicated, or predicted neighbor input.

Let $A_{ij,j}^{\rm ext}$ be the block of $A_{ij}^{\rm ext}$ multiplying $w_j$, and $S_j$ the diagonal block of $S$ corresponding
to $w_j$. If the normalized mismatch satisfies the known bound
$\|\bar w_j-\bar w_j^{\rm ref}\|\le\varepsilon_{ij}$, where
$\bar w_j:=S_j^{-1}w_j$ and
$\bar w_j^{\rm ref}:=S_j^{-1}w_j^{\rm ref}$, then one may take $\rho_{ij}^{\rm comm}(z)
:=
\|A_{ij,j}^{\rm ext}(z)S_j\|\,\varepsilon_{ij}$. Indeed,
$\eta_{4,ij}^{\mathrm{ref}}(z,w_i,w_j)
\geq
\eta_{4,ij}^{\mathrm{ref}}(z,w_i,w_j^{\mathrm{ref}})
-\rho_{ij}^{\mathrm{comm}}(z)$. Thus, if the mismatch bound holds, the local constraint implies the robust
pairwise HOCBF condition for the neighbor input; if $w_j^{\rm ref}=w_j$, $\rho_{ij}^{\rm comm}=0$.

\section{Simulation Study}
\label{sec:simulations}

We evaluate the proposed controller on a three-quadrotor crossing
maneuver designed to expose the distinction between loss of pairwise
thrust effectiveness and bounded-input HOCBF feasibility. The complete
interaction graph is used, with $d_{\min}=0.8$~m. The vehicles have
$m_i=1.30$~kg, $g_0=9.81$~m/s$^2$, and
$J_i=\operatorname{diag}(0.022,0.022,0.040)$~kg\,m$^2$.
The initial and final positions are
$p_1^0=(-2.5,-0.15,1.5)$, $p_1^f=(2.5,-0.15,1.5)$,
$p_2^0=(2.5,0.15,1.5)$, $p_2^f=(-2.5,0.15,1.5)$,
$p_3^0=(0,-2.5,2.4)$, and $p_3^f=(0,2.5,2.4)$.
A quintic smooth-step reference acts over $t\in[3,7]$~s, giving the
nominal crossing near $t=5$~s, with
$v_i(0)=0$, $R_i(0)=I$, $\Omega_i(0)=0$,
$f_i(0)=m_i g_0$, and $q_i(0)=0$.
The nominal tracking gains are
$k_p=4.5$, $k_v=3.4$, $k_R=7$, and $k_\Omega=0.85$,
while the extended thrust tracker uses $k_{fp}=25$ and $k_{fq}=10$.
We impose $f_i\in[1,25]$~N, $|\nu_i|\le100$~N/s$^2$, and
$\|\tau_i\|_\infty\le2.5$~N\,m. Accordingly,
$S_i=\operatorname{diag}(100,2.5,2.5,2.5)$ is used for the normalized
extended-input effectiveness and QP. The collision HOCBFs use
$\alpha_k(s)=8s$, $k=1,\ldots,4$, and the thrust-state HOCBFs use
$\alpha_{f,1}(s)=\alpha_{f,2}(s)=4s$.
The geometric comparison is performed with
$\Delta_{f,i}=\Delta_{\tau,i}=0$ to isolate singularity and aggregate
feasibility; uncertainty is introduced only in the GP experiment below.
All trajectories are integrated with fourth-order Runge--Kutta using a
$5$~ms step. An animation and additional diagnostics are provided in the
supplementary video at
\url{https://youtu.be/UIqZg36BcXE}.

We compare the nominal geometric tracking controller, the
second-order thrust-only HOCBF-QP of Section~II, and the proposed
fourth-order torque-aware HOCBF-QP. As shown in
Fig.~\ref{fig:sim_main} (a), the nominal maneuver violates the safety
distance, reaching $0.300$~m. The second-order QP becomes infeasible at
$t=4.53$~s, while the agents are still $1.673$~m apart. In contrast,
the proposed controller completes the maneuver with
$\min_{t,(i,j)\in\mathcal E}\|p_{ij}(t)\|
    =0.822~{\rm m}>d_{\min}$. Moreover, the minimum values of
$(\eta_{0,ij},\eta_{1,ij},\eta_{2,ij},\eta_{3,ij})$
over all edges and times are
$(0.035,0.100,0.195,0.495)$, respectively.

Figure~\ref{fig:sim_main} (b) isolates the pairwise input-channel
effect on the same nominal states. For edge $(1,2)$, at the safe-side
configuration $t=4.83$~s, where $d_{12}=0.813$~m, the thrust
effectiveness has decreased to $\gamma_{12}^{\rm th}=5.65\times10^{-2}$, whereas the torque-aware extended effectiveness is $\gamma_{12}^{\rm ext,S}=2.57\times10^{3}$. The latter also remains well above the uniform lower bound
$76.92$ predicted by~\eqref{eq:gamma_ext_uniform_bound}. Thus the loss of thrust
effectiveness does not propagate to the dynamically extended input
channel.

Pairwise nonsingularity alone does not ensure bounded-input
feasibility. At the first infeasible point of the second-order QP,
all three thrust-effectiveness measures remain nonzero,
$(\gamma_{12}^{\rm th},
 \gamma_{13}^{\rm th},
 \gamma_{23}^{\rm th})
= (0.493,1.979,1.949)$, while the exact aggregate margin satisfies
$\Gamma_{\rm th}^{\rm rob}=-2.815<0$; see
Fig.~\ref{fig:sim_robust} (a). This numerically illustrates the
distinction in Remark~2.3 between pairwise input effectiveness and
bounded-input feasibility. For the proposed controller,
$\Gamma_{\rm ext}^{\rm rob}$ remains positive throughout the
maneuver, with $\displaystyle{\min_t\Gamma_{\rm ext}^{\rm rob}(t)=5278.6}$. The actuator requirements remain within the prescribed
bounds:
$f_i\in[10.225,16.976]$~N,
$\max_i|q_i|=13.90$~N/s,
$\max_i|\nu_i|=50.44$~N/s$^2$, and
$\max_{i,k}|\tau_{ik}|=0.441$~N\,m.

Finally, we evaluate the robust controller under smooth,
agent-dependent force and torque perturbations of the form
$\Delta_{f,i}=D_{p,i}p_i+D_{v,i}v_i+b_{f,i}$ and
$\Delta_{\tau,i}=K_{\tau,i}\Omega_i+b_{\tau,i}$.
For each edge, a GP is trained offline on $\delta_{ij}^{(4)}$ using
202 samples of the standardized 40-dimensional pair state $z_{ij}$. The nominal fourth-order
trajectory is sampled every $0.1$~s and each sample is augmented by one
independently perturbed state, with standard deviations
$0.10$~m, $0.12$~m/s, $0.025$~rad, $0.04$~rad/s,
$0.35$~N, and $0.7$~N/s in
$(p,v,R,\Omega,f,q)$, respectively.
We use $k(z,z')
=
k_{\rm RBF}(z,z';\ell=4)+10^{-4}\delta_{zz'}$, with unit RBF amplitude, GP regularization $10^{-8}$,
output normalization, and no hyperparameter optimization.
We set $\beta_{ij}=16$, so
$\rho_{ij}^{(4)}=|\mu_{ij}^{\delta}|+4\sigma_{ij}^{\delta}$.
As shown in Fig.~\ref{fig:sim_robust}(b), along the disturbed
closed loop $\displaystyle{\max_{t,(i,j)}
\frac{|\delta_{ij}^{(4)}-\mu_{ij}^{\delta}|}
{4\sigma_{ij}^{\delta}}
=0.335<1}$, and all evaluated residual errors lie within the $4\sigma$ envelope.
The robust controller maintains $d_{\min}^{\rm obs}=0.895$~m,
$\min(\eta_{0,ij},\eta_{1,ij},\eta_{2,ij},\eta_{3,ij})
=(0.161,1.131,8.189,62.235)$,
$\min\eta_{4,ij}^{\rm true}=245.79$, and
$\min\Gamma_{\rm ext}^{\rm rob}=5528.58$.
The resulting actuator ranges remain
$f_i\in[9.527,18.018]$~N,
$\max|q_i|=17.01$~N/s,
$\max|\nu_i|=65.88$~N/s$^2$, and
$\max|\tau_{ik}|=0.499$~N\,m.

\begin{figure}[t]
    \centering
    \includegraphics[width=0.49\columnwidth]{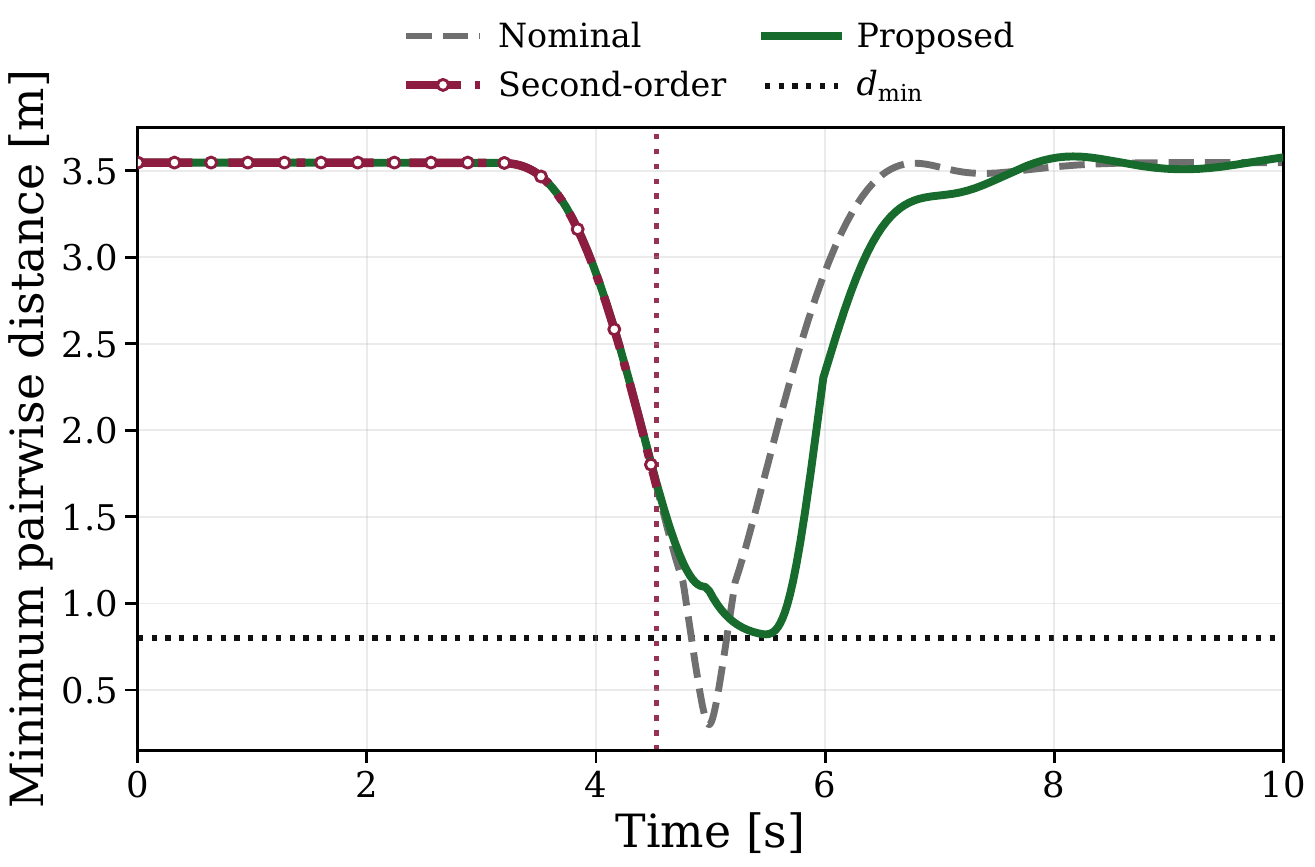}
    \hfill
    \includegraphics[width=0.49\columnwidth]{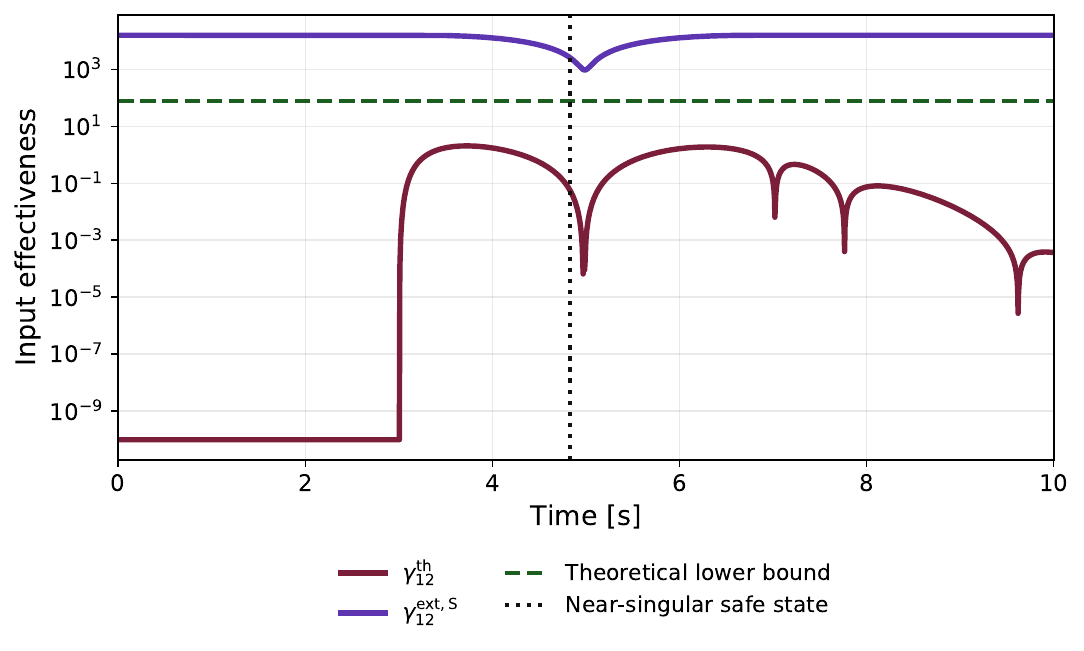}
    \caption{Crossing maneuver.
(a) Minimum pairwise distance for the nominal, second-order thrust-only,
and proposed controllers.
(b) Thrust and extended input effectiveness for edge $(1,2)$ on the same nominal states.}
    \label{fig:sim_main}
\end{figure}

\begin{figure}[t]
    \centering
    \includegraphics[width=0.49\columnwidth]{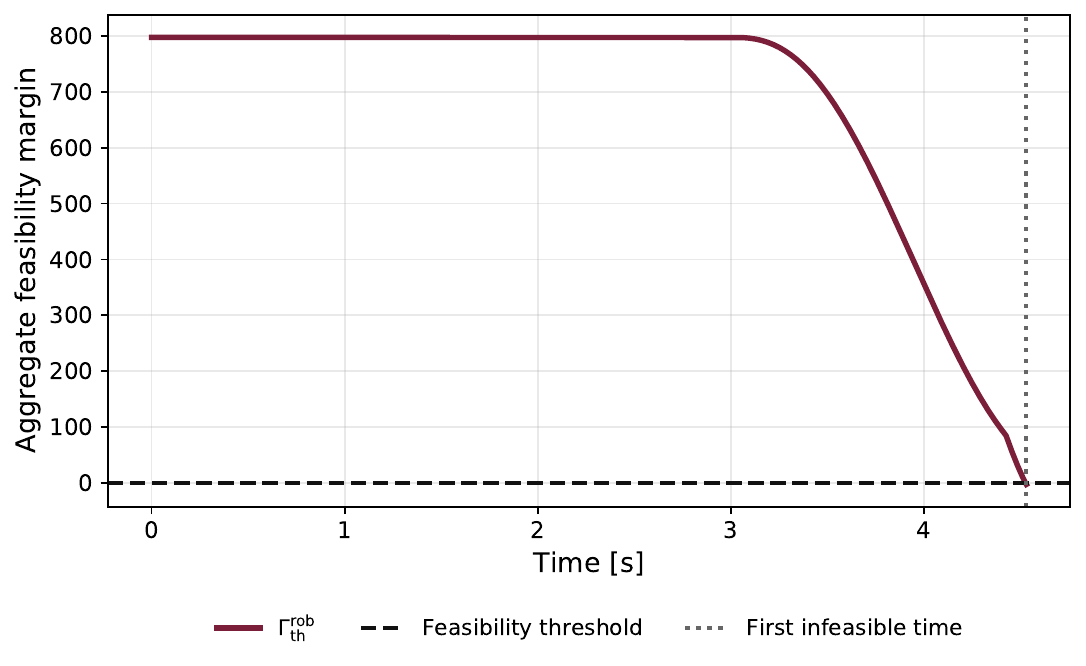}
    \hfill
    \includegraphics[width=0.49\columnwidth]{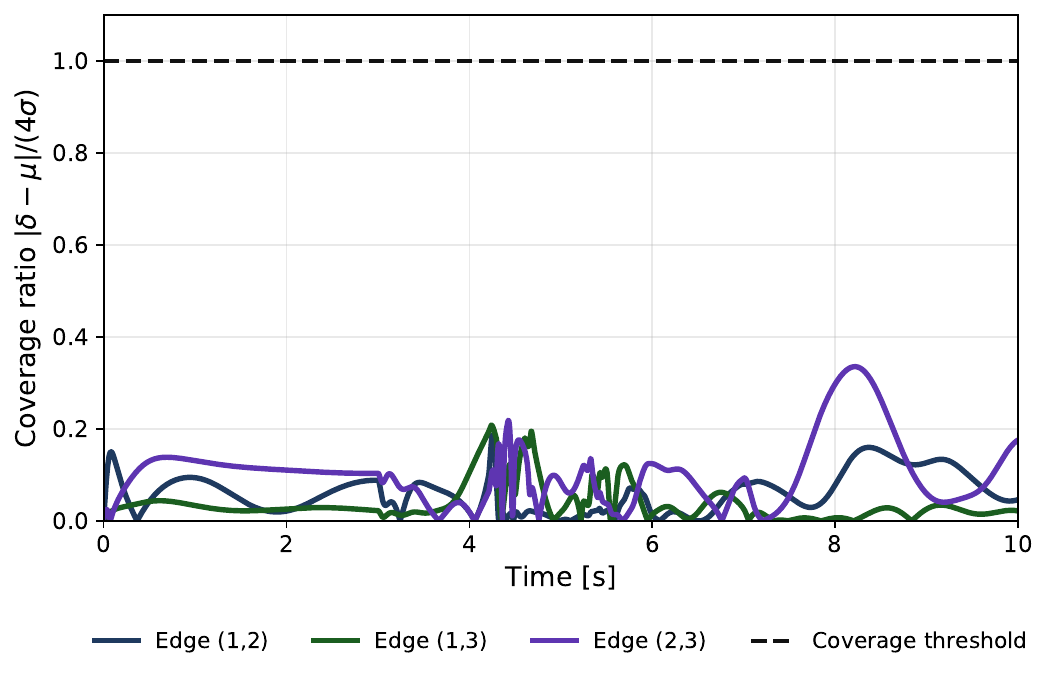}
    \caption{Aggregate feasibility and GP robustness.
(a) Exact thrust-feasibility margin along the second-order closed loop.
(b) Empirical residual-envelope coverage for all edges.}
    \label{fig:sim_robust}
\end{figure}

\section{Conclusions}
We studied pairwise singularity and aggregate feasibility in
safety-critical control of quadrotor teams under bounded inputs and
model uncertainty, characterized the thrust singular set, and derived
an exact aggregate feasibility margin for shared bounded-thrust
constraints. A torque-aware dynamic extension exposes thrust
acceleration and attitude torques in a fourth-order HOCBF, preventing
pairwise extended-input degeneracy under positive thrust bounds. Robustness is
obtained by directly learning the fourth-order HOCBF residual with GPs,
while normalized input-effectiveness measures provide local feasibility
conditions for the extended-input QP. Future work will address
experimental multi-UAV validation and less conservative
aggregate-feasibility guarantees.

\bibliographystyle{IEEEtran}
\bibliography{autosam}

\end{document}